\documentclass[aps,prl,twocolumn,notitlepage,superscriptaddress]{revtex4-2} 

\usepackage{qcircuit}

\usepackage[linesnumbered,ruled,vlined]{algorithm2e}
\SetAlgorithmName{Protocol}{Protocol}{Protocol}
\usepackage{algpseudocode}

\usepackage{epsfig,amssymb,amsmath,mathrsfs,amsfonts,color,amsthm,graphicx,float,color,xcolor}
\usepackage[colorlinks=true]{hyperref}
\usepackage[active]{srcltx}
\usepackage{subfigure}
\def\<{\langle}\def\>{\rangle}

\def\Tr{\operatorname{Tr}}\def\:{\hbox{\bf
    :}}
\def\diag{\operatorname{diag}}

\def\N{\mathbb N}

\def\Z{\mathbb Z}

\def\set#1{\mathsf{#1}}

\def\GHZ{\mathrm{GHZ}}
\newcommand*{\YY}[1]{{\color{cyan} [YY: #1]}}

\newcommand*{\ZP}[1]{{\color{blue} [ZP: #1]}}

\DeclareMathOperator\tr{Tr}

\def\id{\mathbb{I}}%
\def\cov{\mathrm{Cov}}
\def\var{\mathrm{Var}}
\def\ot{\leftarrow}%
\newtheorem{theo}{{Theorem}}
\newtheorem{lemma}{{Lemma}}

\begin{document}
\title{Quantum-enhanced metrology with network states}
 
 \author{Yuxiang Yang}
 \affiliation{QICI Quantum Information and Computation Initiative, Department of Computer Science,
The University of Hong Kong, Pokfulam Road, Hong Kong, China}
 \author{Benjamin Yadin}
 \affiliation{Naturwissenschaftlich-Technische Fakult\"{a}t, Universit\"{a}t Siegen, Walter-Flex-Straße 3, 57068 Siegen, Germany}
 \author{Zhen-Peng Xu}
 \email{zhen-peng.xu@ahu.edu.cn}
 \affiliation{School of Physics and Optoelectronics Engineering, Anhui University, Hefei 230601, People's Republic of China}
\begin{abstract}
Armed with quantum correlations, quantum sensors in a network have shown the potential to outclass their classical counterparts in distributed sensing tasks such as clock synchronization and reference frame alignment. On the other hand, this analysis was done for simple and idealized networks, whereas the correlation shared within a practical quantum network, captured by the notion of network states, is much more complex.
Here, we prove a general bound that limits the performance of using quantum network states to estimate a global parameter, establishing the necessity of genuine multipartite entanglement for achieving a quantum advantage. The bound can also serve as an entanglement witness in networks and can be generalized to states generated by shallow circuits.
Moreover, while our bound prohibits local network states from achieving the Heisenberg limit, we design a probabilistic protocol that, once successful, attains this ultimate limit of quantum metrology and preserves the privacy of involved parties. Our work establishes both the limitation and the possibility of quantum metrology within quantum networks.
\end{abstract}
\maketitle

\noindent{\it Introduction.} Distributed sensing in a network is a general task of fundamental significance. Remarkably, with Greenberger-Horne-Zeilinger (GHZ) states over a network of $M$ parties, it is possible to estimate a global parameter $\theta$ with mean squared error $\Delta^2(\theta)\sim 1/M^2$, achieving the Heisenberg limit of quantum metrology --- with a $\Theta(M)$ reduction of error over protocols without entanglement \cite{giovannetti2004quantum}.
This fact lies behind the recent interest in  distributed quantum sensing
\cite{komar2014quantum,proctor2018multiparameter,ge2018distributed,eldredge2018optimal,zhuang2018distributed,qian2019heisenberg,sekatski2020optimal,zhang2021distributed,fadel2022multiparameter}
including several experimental demonstrations 
\cite{guo2020distributed,xia2020demonstration,liu2021distributed,zhao2021field} in both finite dimensional systems and continuous-variable systems.  

In practice, however, the distribution of global entangled states, e.g., multipartite GHZ states, is a daunting task~\cite{zhong2021deterministic,zhang2022quantum}. The decoherence time of multipartite entanglement leads to experimental limits in transmission and storage~\cite{carvalho2004decoherence}, especially for remote parties.
A feasible solution is to use quantum repeaters~\cite{briegel1998quantum} and consider distributed settings, i.e., quantum networks~\cite{kimble2008quantum,ladd2010quantum}.
Most often the generic resource states are \emph{network states}
~\cite{Navascues2020Genuine,aaberg2020semidefinite,luo2021new,kraft2021quantum,kraft2021characterizing,hansenne2022symmetries}, prepared by distributing few-partite entangled sources to different vertices and applying local operations according to a predefined protocol.
On the other hand, typical resources for quantum metrology like (multipartite) GHZ states and graph states ~\cite{audenaert2005entanglement,hein2006entanglement} may not be accessible in generic networks ~\cite{hansenne2022symmetries,makuta2022no,wang2022quantum}. 
Consequently, it is natural to ask how to characterize the potential of network states in sensing global parameters, and whether the Heisenberg limit can still be achieved. These critical questions, however, have been largely unexplored due to the much more general and complex nature of network states compared to GHZ states. 

In this work, we derive a versatile general upper bound on the precision of any deterministic protocol for estimating a (global) parameter using network states, which leads to sufficient conditions under which the  precision is bounded by the standard quantum limit (SQL) $\Delta^2(\theta)\sim 1/M$ (for $M$ parties) and cannot achieve the Heisenberg limit (HL) $\Delta^2(\theta)\sim 1/M^2$. We then design a probabilistic sensing protocol \cite{gendra2013quantum,gendra2014probabilistic,combes2014quantum,arvidsson2020quantum,lupu2022negative,xiong2022advantage,chiribella2013quantum,calsamiglia2016probabilistic} using local post-selection to achieve the HL, which also features the preservation of local parameters' privacy. 
 

\medskip

\noindent{\it Distributed sensing with network states.} 
A quantum network state $\rho$ is a multipartite quantum state, whose structure can be efficiently represented by a hypergraph $G(\mathcal{V},\mathcal{E})$ of $K(=|\mathcal{V}|)$ vertices and $|\mathcal{E}|$ hyperedges (i.e., subsets of $\mathcal{V}$). Each vertex represents a local site, and each hyperedge represents an entanglement source. 
The network state $\rho$ is generated via a two-step procedure, where each entanglement source, represented by a hyperedge $e$, distributes an entangled state to every local site in $e$ and then each local site $v$ applies an arbitrary local operation. 
Moreover, all sites can be classically correlated via a pre-shared global random variable $\lambda$.
In other words, a network state is of the generic form  
\begin{equation}\label{eq:nolosr}
    \rho = \sum_\lambda p_\lambda \rho_\lambda,\ \rho_\lambda = \left(\bigotimes_{v\in \mathcal{V}} \Phi_{v}^{(\lambda)} \right) \left(\bigotimes_{e\in \mathcal{E}} \sigma_{e}\right),
\end{equation}
where $\Phi_{v}^{(\lambda)}$ is a channel acting on sensor $v$, and $\sigma_e$ is an entangled state shared between sensors $v \in e$.

The goal of distributed sensing is for a group of far-apart sensors, each having access to an unknown local signal, to estimate one (usually global) parameter, e.g., the average of all local parameters \footnote{Note that in general there could be more than one parameter of interest, while in this work we focus on the fundamental case of estimating only one parameter.}. To accomplish this goal, the local sensors have access to a joint network state in each round of the experiment. 
In general, a \emph{deterministic protocol} consists of three phases:
$i)$ {\em Network state distribution:} Each entanglement source $e$ distributes an entangled state (e.g., Bell pairs or GHZ states) among its associated local sites (sensors) $\{v\in\mathcal{V}~|~v\in e\}$. Each sensor performs a local operation. Eventually, the sensors share a network state $\rho$.
$ii)$ {\em Signal acquisition:} The sensors obtain local signals. Explicitly, the state goes through a unitary evolution $U(\vec{\theta})=\exp\{-i\sum_{s\in\mathcal{S}}H_s\theta_s\}$, where $\{\theta_s\}$ are unknown parameters with generators $\{H_s\}$ and $\mathcal{S}$ consists of subsets of $\mathcal{V}$. We denote by $M:=|\mathcal{S}|$ the cardinality of $\mathcal{S}$. A generator $H_s$ acts trivially on a sensor $v$ if $v\not\in s$. The global state becomes $\rho(\{\theta_s\})$ after signal acquisition.
$iii)$ {\em Parameter estimation:} The parameter of interest is a function $f(\{\theta_s\})$ of $\{\theta_s\}$. The form of $f$ is known to all sensors. Depending on this function, a measurement is performed on the global state $\rho(\{\theta_s\})$ and an unbiased estimate $\hat{\theta}$ of $\theta$ is extracted from the measurement outcome statistics. Note that, in practice, the type of measurements that could be performed is often restricted (e.g., to be local). Here we show a stronger result (Theorem \ref{thm-generalbound}) that holds without any constraint on what kind of measurements may be performed.

For example, Figure~\ref{fig:cyclic} a) shows a network of $K=3$ sensors ($v_1,v_2,v_3$) with a shared bipartite entangled state ($\mathcal{E}=\{\{v_1,v_2\},\{v_1,v_3\},\{v_2,v_3\}\}$) between each pair of sensors. There are $M=3$ parameters, each collected by an individual sensor ($\mathcal{S}=\{\{v_1\},\{v_2\},\{v_3\}\}$).
\begin{figure}
\includegraphics[width=\linewidth]{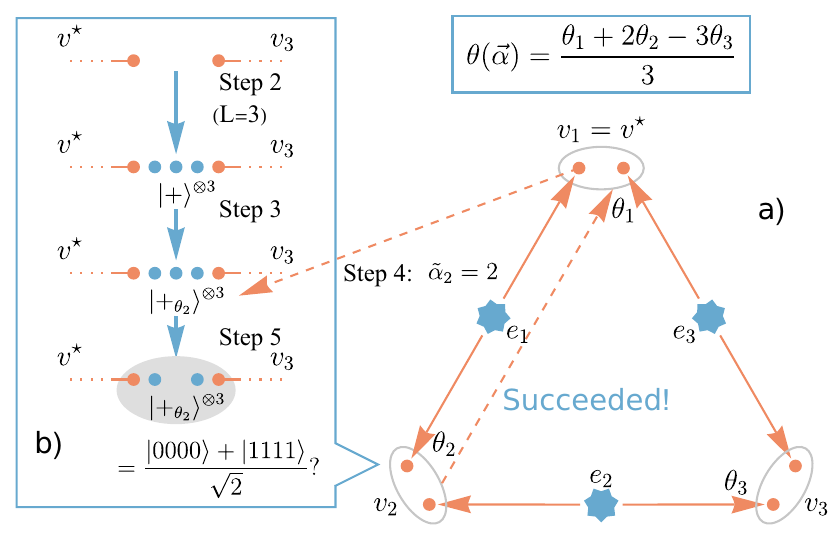}
\caption{{\bf a) A cyclic network consisting of $M=3$ sensors.} Here each pair of sensors share a Bell state $|B_0\>=\frac{|00\>+|11\>}{\sqrt{2}}$ via a source.  {\bf b) Steps 2-5 of Protocol \ref{protocol-general}} are illustrated for one of the non-center sensor, $v_2$.}\label{fig:cyclic}
\end{figure}

\medskip

\noindent{\it A general lower bound on the estimation error.}  Consider the task of estimating a global parameter that is an arbitrary linear combination of the local parameters $\theta(\vec{\alpha}):=\vec{\alpha}^T\vec{\theta}=\sum_s \alpha_s\theta_s$,
where $\vec{\alpha}$ is a vector of dimension $M$. \footnote{To estimate $f(\{\theta_s\})$ that is non-linear in the parameters, one can assume w.l.o.g. that the parameters lie in the vicinity of the true values (see, e.g., \cite[Section 9]{yang2019attaining}) and reduce to the linear case via Talyor-expansion.}
The mean squared error can be expressed as
\begin{align}\label{mse-theta}
\Delta^2\left(\theta(\vec{\alpha})\right)=\vec{\alpha}^T\cov\left(\{\hat{\theta}_s\}\right)\vec{\alpha},
\end{align}
where $\cov\left(\{\hat{\theta}_s\}\right)$ is the covariance matrix of the estimators $\{\hat{\theta}_s\}$ for $\{\theta_s\}$: $\left(\cov\right)_{ij}:=\mathbb{E}[(\hat{\theta}_i-\theta_i)(\hat{\theta}_j-\theta_j)]$ for $1\le i,j\le M$.

Our first main result is a general lower bound on the error of any deterministic protocol for an arbitrary network state. {Compared with works in distributed sensing \cite{komar2014quantum,proctor2018multiparameter,ge2018distributed,eldredge2018optimal,zhuang2018distributed,qian2019heisenberg,sekatski2020optimal,zhang2021distributed}, our bound captures not only the impact of the sensing task, but also the architecture of the network. More explicitly, we identify a key quantity in determining the error scaling, named the \emph{influence of the local signal $s$:} $k_s := \max_{e\in \mathcal{E}, e\cap s\neq \emptyset} |\{t \in \mathcal{S} |t\cap e\neq \emptyset\}|$. Intuitively, $k_s$ is the maximum number of local signals influenced by $s$ via an entanglement source within the network. Note that $k_s$ depends on both the sensing task and the network.
}
\begin{theo}\label{thm-generalbound}
When estimating a global parameter $\theta(\vec{\alpha})=\vec{\alpha}^T\vec{\theta}$ given a network state $\rho$ with structure $G(\mathcal{E},\mathcal{V})$,
the mean squared error of any deterministic protocol is lower bounded as:
\begin{align}\label{eq:variance_ineq}
   \Delta^2\left(\theta(\vec{\alpha})\right)\ge \sum_s \frac{\alpha_s^2}{4\nu k_s \var(\rho,H_s)}.
\end{align}
Here $H_s$ is the generator of $\theta_s$, $\var(\rho,H_s)$ is the variance of $H_s$ with respect to $\rho$, and $\nu$ is the number of rounds that the experiment is repeated.
\end{theo}

In principle, the bound (\ref{eq:variance_ineq}) can be extended to the more general multiparameter case where the cost function is of the form $\Tr(W\cov(\{\hat{\theta}_s\}))$ for some weight matrix $W\ge0$. In fact, we prove the bound by combining the (matrix) Cram\'er-Rao bound \cite{helstrom1969quantum,holevo2011probabilistic} $\cov(\{\hat{\theta}_s\})\ge(1/\nu)\mathcal{F}_Q^{-1}$ with the following bound on the quantum Fisher information (QFI) matrix (see Appendix for its explicit definition) for the parameters $\{\theta_s\}$:
\begin{align}\label{QFI-bound}
   \mathcal{F}_Q(\rho,\{H_s\}) \le \diag \left\{4 k_s \var(\rho, H_s) \right\}_s,
\end{align}
where $\rho$ denotes the network state, and $\mathcal{F}_Q(\rho,\{H_s\})$ denotes the quantum Fisher information matrix of the state   $\rho_{\vec{\theta}}:=U(\vec{\theta})\rho U(\vec{\theta})^\dag$ with $U(\vec{\theta}):=e^{-i\sum_s H_s\theta_s}$. Its proof can be found in the Appendix. 

We now check how the error scales with respect to $M$, the number of local parameters. In Eq.~(\ref{eq:variance_ineq}), the variance $\var(\rho,H_s)$ can be bounded as $\var(\rho,H_s)\le\|H_s\|^2\le h_{\max}^2$ (with $\|\cdot\|$ being the operator norm), where $h_{\max}:=\max_s\|H_s\|$ is $M$-independent. 
For estimating the mean of $\{\theta_s\}$, we have $\vec{\alpha}=(1/M,\dots,1/M)^T$, and thus $\vec{\alpha}^T\vec{\alpha}\sim 1/M$. 
Further, denoting by $k_{\max}$ the maximum of the influence $k_s$, Eq.~(\ref{eq:variance_ineq}) implies
\begin{align}\label{eq:reduced_bound}
   \Delta^2\left(\theta(\vec{\alpha})\right)= \Omega\left(\frac{1}{M\cdot k_{\max}\cdot h_{\max}^2}\right).
\end{align}

A key implication of our result is that genuine $M$-partite {\emph{network}} entanglement is necessary for achieving the HL in a network. 
{When $\Delta^2\sim 1/M^2$, Eq.~(\ref{eq:reduced_bound}) requires $k_{\max}$ to scale with $M$.}
As long as $s\cap s'=\emptyset$ for any pair of local signals, $k_{\max}$ is upper bounded by $\max_{e\in\mathcal{E}}|e|$, which captures the range of genuine entanglement in the network \cite{Navascues2020Genuine}. 
{Our result thus establishes a crucial connection between this core property of a generic network and quantum metrology, requiring it to scale with $M$ to achieve the HL.
Since network entanglement is a stronger and more natural resource for network scenarios than multipartite entanglement \cite{Navascues2020Genuine}, our result extends the main result of Ref.~\cite{ehrenberg2023minimum}, where the necessity of $M$-partite entanglement was established.}
Our bound also shows that local pre-processing cannot be used to gain an advantage in metrology with shared entangled states between limited numbers of parties.

{Furthermore, it is revealed by our bound that not only the amount of entanglement but also the \emph{architecture of the network} is important, when analysing each influence $k_s$ instead of $k_{\max}$. As an example, consider the task shown in Fig.~\ref{fig:sun}, where the state processes an amount of genuine network entanglement that scales with $M$ but fails to achieve the HL. A sufficient condition for the bound (\ref{eq:reduced_bound}) to attain the HL is to have all vertices covered by $M$-hyperedges. Constructing a sensing protocol for such network states is an interesting  direction for  future work.}

\begin{figure}[bt]
\includegraphics[width=0.8\linewidth]{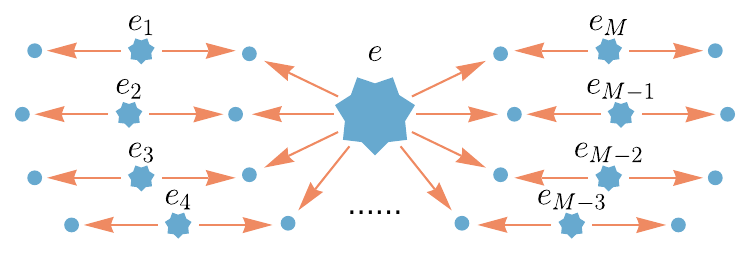}
\caption{{\bf Genuine network entanglement of order $M$ is not sufficient for Heisenberg scaling.} Consider a task of estimating the average of $2M$ local parameters in a network, which consists of a hyperedge $e$ with cardinality $M$ and $M$ edges $e_1,\dots,e_M$. The network state has genuine $M$-partite network entanglement, but the squared error scales as $1/M$ by Eq.~(\ref{eq:reduced_bound}) since there are $M$ vertices with $k_s=2$.}\label{fig:sun}
\end{figure}

\medskip

\noindent{\it Error bound as an entanglement witness.}
The bound (\ref{QFI-bound}) on $\mathcal{F}_Q$ (and, equivalently, Theorem \ref{thm-generalbound}) can also be used as a witness of genuine  multipartite network entanglement~\cite{Navascues2020Genuine}.
Previous works have shown that QFI with respect to a non-interacting Hamiltonian can detect genuine $k$-body entanglement, which can not be produced by probabilistic mixing of pure states in which no more than $(k-1)$ parties are entangled~\cite{Toth2012Multipartite,Hyllus2012Fisher}, and that an interacting Hamiltonian can be used to rule out fully separable states~\cite{Pezze2015Witnessing}.
Our result extends both types of results to cover genuine network $k$-entanglement~\cite{Navascues2020Genuine}, defined as entanglement that cannot be generated in a network in which sources are distributed to at most $k-1$ parties at a time.
Due to the entanglement distribution and the inclusion of the parties' local channels, witnessing genuine network $k$-entanglement is stronger than witnessing genuine $k$-body entanglement. 

For example, consider a one-dimensional spin model with nearest-neighbor coupling: $H = \sum_{i=1}^M H_i$ such that each $H_i$ acts nontrivially only on sites $i$ and $i+1$.
In terms of a network, each vertex is a site $i$ and each $s$ is a pair $(i,i+1)$.
If the state can be prepared from $r$-partite sources, then it is easily seen that $k_s \leq 2r$, where the upper bound is achieved when the particles from the same source do not interact with each other. 
For such a state $\rho$, Theorem \ref{thm-generalbound} yields
\begin{align} 
   \Delta^2\left(\theta(\vec{\alpha})\right)\ge \sum_{i=1}^M \frac{\alpha_i^2}{8\nu r  \var(\rho,H_i)}.
\end{align}

If the model is translationally invariant and $\alpha_i = 1/M$ for every $i$ (i.e., to estimate the average $\bar{\theta}$), the above bound is reduced to
$\Delta^2\left(\bar{\theta}\right)\ge 1/(8\nu r M \var(\rho,H_1))$.
More generally, for a model with bipartite interactions and $\tau$ nearest neighbors per site, we see that $\Delta^2\left(\bar{\theta}\right)\ge 1/[4 \nu \tau r M \var(\rho,H_1)]$; 
for instance, a $d$-dimensional cubic lattice has $\tau = 2^d$. 
The case of an Ising model was studied in Ref.~\cite{Pezze2015Witnessing} for $r=1$, taking $H_i = \frac{1}{2}Z_i + \frac{\epsilon}{4} Z_i Z_{i+1}$.
By optimizing over pure fully separable states, they find the upper bound on the QFI to be $M(1 + \frac{5 \epsilon^2}{4})$ for $\epsilon \ll 1$ and $M(\frac{1}{2} + \epsilon + \frac{\epsilon^2}{2})$ for $\epsilon > \epsilon_c \approx 0.7302$.
Our bound gives $M(2+2\epsilon + \frac{\epsilon^2}{2})$, which is less tight, but has the advantage of being easily extended to $r>1$ without the need to search for optimal states.

\medskip

\noindent{\it Precision bounds for shallow circuits.}
Our technique can also be applied to scenarios where an entangled probe state is prepared from a circuit with local gates of a finite depth.
We first consider an unentangled (i.e., fully separable) state $\rho$ input into a circuit composed of $l$-local gates.
This has depth $D$, meaning that there are $D$ layers such that the gates are non-overlapping within each layer.
The gates have corresponding unitaries $U_{j,\alpha}$, where $j$ labels the layer and $\alpha$ the index within the layer.
The unitary for the full circuit is $U = (\otimes_{\alpha_D} U_{d,\alpha_D}) \dots (\otimes_{\alpha_i} U_{1,\alpha_1})$.
The output from the circuit $\sigma = U\rho U^\dagger$ is used as a probe for sensing rotations generated by a $p$-local Hamiltonian $H = \sum_{i=1}^M H_i$ where each $H_i$ acts locally.
Then the QFI is $\mathcal{F}_Q(\sigma,H) = \mathcal{F}_Q(\rho, U^\dagger H U)$.
The locality $q$ of the transformed Hamiltonian $U^\dagger H U$ can be bounded using the depth $D$.
We can give different bounds depending on the structure of the circuit.
This follows from a light cone argument as described in Ref.~\cite{Yu2023Learning}.

With no particular geometry, the weight of each Hamiltonian term (i.e., number of subsystems acted upon) increases by a factor of at most $l$ under the application of each layer.
Therefore, we have $q \leq l^D p$. A more useful bound cannot be obtained unless we know the circuit structure. 
If, for example, the circuit forms a one-dimensional chain with $l=2$ and $H$ being 1-local, then we see that $U^\dagger H U$ has terms interacting $q \leq 2D +1$ neighbors.
It follows that $\mathcal{F}_Q(\sigma,H) \leq 4 q \sum_{i=1}^M \var(\rho, U^\dagger H_i U) \leq 4(2D+1) \sum_{i=1}^M \var(\rho, U^\dagger H_i U)$.
The HL can therefore only be approached if $D \sim M$.
A two-dimensional square lattice works similarly, replacing $q \leq D^2 + (D+1)^2$ by counting the number of lattice points reachable from a given starting point in $D$ unit steps.
Hence, we now only need a shallow circuit with $D \sim \sqrt{M}$ to get the HL.
In general, with local geometry, we would require $D$ to be similar to the size of the system. 
Using this idea, we can also bound the QFI when the parameter of interest is embedded in the circuit $U$ (see Appendix for details).

\medskip

\noindent{\it Achieving the Heisenberg limit via post-selection.} 
Next we introduce a concrete protocol.
Consider a network state with a hypergraph topology $G(\mathcal{V},\mathcal{E})$,  where each $e\in \mathcal{E}$ represents a GHZ state shared by all sensors $v\in e$.
Each sensor $v_j$ can locally access a signal $\theta_j$, which is a phase gate $e^{-i\theta_jZ_j/2}$ with $Z_j$ being the Pauli-$Z$ operator located at the $j$-th sensor's place. That is, $\mathcal{S}$ is the collection of all singletons and $|\mathcal{S}|=|V|$. 
The task is for a center, which could be any one of the sensors, to estimate an arbitrary combination of parameters
$\theta(\vec{\alpha}):=\frac{1}{M}\sum_{j=1}^M \tilde{\alpha}_j\theta_j$ 
with the assistance of all sensors, where with every $\tilde{\alpha}_j\in\Z\setminus \{0\}$ and $|\tilde{\alpha}_j|$ is upper bounded by $L$ for some known constant $L\in\N^*$. 

Besides the precision of estimation, we require the parameter estimation to respect the privacy of involved parties. Concretely, the center requires no other party to obtain the full knowledge of $\vec{\alpha}$. On the other hand, each sensor wishes to keep its local value $\theta_j$ private while still assisting the center to estimate $\theta(\vec{\alpha})$.

The standard deterministic protocol of measuring locally and communicating the outcomes to the center fails the second requirement unless the network state is of specific global topology (e.g., being a large GHZ state spanning all vertices) \footnote{The privacy of the local parameters may be preserved by a modification of the deterministic protocol: The $j$-th sensor first estimates $\theta_j$ and then sends a protected phase $\Theta_j=\theta_j+v_j$ to the center with some randomly generated value $v_j$. Provided $\vec{v}^T\vec{\alpha}=0$, the center can effectively learn $\theta(\vec{\alpha})=\vec{\alpha}^T\vec{\Theta}$ without knowing $\vec{\theta}$. However, to distribute to each sensor the desired $v_j$, it requires a trusted third party who learns $\vec{\alpha}$, violating the first privacy requirement.}.
We present a probabilistic protocol that achieves both privacy requirements and, as a bonus, attains HL when successful, with almost no restriction on the network topology.
Probabilistic protocols are prevalent in quantum sensing \cite{mitchell2004super,walther2004broglie,nagata2007beating} and quantum information processing \cite{gendra2013quantum,gendra2014probabilistic,combes2014quantum,arvidsson2020quantum,lupu2022negative,xiong2022advantage,chiribella2013quantum,calsamiglia2016probabilistic}. 
For metrology, it has been shown to boost the precision to the HL \cite{gendra2013quantum} and beyond \cite{arvidsson2020quantum}. Whether such an appealing feature persists under the constraint of quantum networks, however, remains largely unexplored. 
The protocol runs as described by Protocol \ref{protocol-general}.
Note that, if $\vec{\alpha}$ is known \textit{a priori}, we need assume only that the sensors have universal local control and the classical communication of measurement outcomes to the center can be delayed until the end. Therefore the sensors also do not need a quantum memory. However, if the ``task allocation'' step is included, the choice of $\vec{\alpha}$ is communicated to the sensors before the measurement step.

\begin{algorithm}[h]
  \caption{Probabilistic metrology of a global parameter over a network state.}
  \label{protocol-general}
   \begin{algorithmic}[1]
    \State (State preparation.) Each source $e$ prepares a $|e|$-qubit GHZ state and distributes it among $v\in e$.
    \State (Center election and local pre-processing.) An arbitrary sensor $v^\ast$ is selected as the center. Each of the other sensors locally prepares $L$ copies of the plus state $|+\>:=(1/\sqrt{2})(|0\>+|1\>)$.
    \State (Signal acquisition.) Each $v_j\not=v^\ast$ passes each plus states once through the signal $\theta_j$. The resultant state is (up to an irrelevant global phase) $|+_{\theta_j}\>^{\otimes L}$ with $|+_{\theta_j}\>:=(1/\sqrt{2})(|0\>+e^{i\theta_j}|1\>)$.
    \State (Task allocation.) The center sends the weight $\tilde{\alpha}_j$ to the $j$-th vertex.
    \State ({Local} measurement and post-selection.) For each sensor $v_j$ (including $v^\ast$), it first discards $L-|\tilde{\alpha}_j|$ plus states and keeps the state $|+_{\theta_j}\>^{\otimes\tilde{\alpha}_j}$.  
    Next, if $\tilde{\alpha}_j<0$, the sensor performs $X$ on each qubits of the plus states. If $v_j=v^\ast$ in this case, it performs $X$ on the single qubit that acquired the signal.
    Finally, the sensor performs a binary measurement $\{|\GHZ_j\>\<\GHZ_j|,I-|\GHZ_j\>\<\GHZ_j|\}$, where $|\GHZ_j\>$ denotes the GHZ state of all its local qubits (including the plus states with signals and the qubits from the sources).  For every $v_j\not=v^\ast$, it declares success if the first outcome (i.e., $|\GHZ_j\>\<\GHZ_j|$) is obtained.
    \State (Estimation.) Conditioning on all other sensors declaring success, the probability that $v^\ast$ yields $|\GHZ^*\>\<\GHZ^*|$ is a function of $\theta$, from which an unbiased estimate $\hat{\theta}$ on $\theta(\vec{\alpha})$ can be obtained.  
    \end{algorithmic}
\end{algorithm} 

For Protocol \ref{protocol-general}, it is obvious that the first privacy requirement is fulfilled, as each sensor learns only a single entry $\tilde{\alpha}_j$. Note that in the (most interesting) scenario of a large network, it can be extremely difficult for the sensors to conspire and communicate the entire $\vec{\alpha}$.
In addition, the probability that an arbitrary clique of sensors succeeds (and other sensors fail) is independent of the values of $\{\theta_i\}$. The conditional state at the center also depends only on $\theta(\vec{\alpha})$ (see Appendix for the proof). As a consequence, no additional information on the local parameters $\{\theta_i\}$ is leaked to the center whether or not the protocol is successful. The second requirement of privacy is thus fulfilled.

As for the precision, Protocol \ref{protocol-general} 
makes $ML$ queries in total to the local signal sources for each run (when estimating the average, in particular, we have $L=1$). When $\{\tilde{\alpha}_j\}$ are (non-integer) rational numbers, one may still apply the protocol with $\tilde{\alpha}_j$ replaced by $\tilde{L}\tilde{\alpha}_j$, where $\tilde{L}$ is the smallest positive integer such that $\{\tilde{L}\tilde{\alpha}_j\}$ are all integers. The queries per run becomes $M\tilde{L}\max_j\tilde{\alpha}_j$. For a generic irrational $\tilde{\alpha}_j$, one may round it to a close rational number prior to applying the protocol. 
The protocol achieves the HL (i.e., QFI$\sim M^2$) as long as $G$ remains connected after removing the center (i.e., when $v^\ast$ is not a cut-vertex). The proof can be found in Appendix.

With this, we see that {combining local resources and post-selection allows us to achieve the optimal sensing precision up to a constant.}
In contrast, deterministic protocols are bound by the standard quantum limit (i.e., QFI $\sim M$) by Theorem \ref{thm-generalbound}. Note that the few-partite entanglement in network states plays an essential role in achieving the (probabilistic) HL, as a completely local state remains local after post-selection.
As a concrete example, one may consider $G$ being a cyclic network of $M$ vertices, where each vertex (sensor) $v_j$ holds a local signal with parameter $\theta_j$ (see Figure \ref{fig:cyclic} a)), and each edge represents one Bell state
$|B_0\>:=\frac{1}{\sqrt{2}}(|00\>+|11\>)$. 
This state is obviously local, and by Theorem \ref{thm-generalbound} the performance of any deterministic protocol is bounded by the standard quantum limit. On the other hand, Protocol \ref{protocol-general}, when applied to this cyclic network (see Figure \ref{fig:cyclic} b) for an illustration), achieves the HL. 

The overall success probability of Protocol \ref{protocol-general} can be 
  lower bounded as
(see Appendix for details)
\begin{equation}
    \log_2 p_{{\rm succ}}  \ge  -\Big[\sum_{v\neq v^\ast} |\tilde{\alpha}_v| + \sum_{e\in \mathcal{E}} |e| - |\mathcal{E}(v^\ast)| \Big],
\end{equation}
with $\mathcal{E}(v^\ast)$ being the set of edges containing $v^\ast$ and $|e|$ being the number of vertices contained in $e$.
The probability $p_{{\rm succ}}$ vanishes with increasing network size.  The advantage thus vanishes if one takes into account the failed cases, which is a general limitation of post-selected metrology \cite{combes2014quantum} rather than a defect of Protocol \ref{protocol-general}.

Protocol \ref{protocol-general} also features other interesting advantages. One is that the sensors could decide which $\theta(\vec{\alpha})$ to estimate after each local signal is applied. Via post-selection, they could ``steer" the state to maximize the precision of $\theta(\vec{\alpha})$ for one particular configuration of $\vec{\alpha}$. Whenever needed, the election of the center can even be delayed to after the acquisition of the signal by a minor modification of the protocol (i.e., to let every sensor prepare $\frac{|0\>+|1\>}{\sqrt{2}}$ states to host the signal). We emphasize that no extra classical communication (at the ``task allocation" step) is needed if  $\vec{\alpha}$ is known \textit{a priori}. Therefore, the extra classical communication is not the reason for the precision enhancement.

Existing works of cryptographic quantum metrology required global entanglement (e.g., GHZ states spanning the network) \cite{yin2020experimental,shettell2022cryptographic,shettell2022quantum,shettell2022private,kasai2023anonymous}, while our work shows the possibility of achieving desired privacy when global entanglement is not accessible. It should also be noted that our discussion does not cover the verification of probe states, which can be discussed independently.

\medskip

\noindent{\it Conclusion.} In this work, we explored both the limitation (by deriving a general bound) and the potential (by designing a probabilistic protocol) of network states in metrology. 
Our work opens a new line of research and hints at more future possibilities along it. For example, it is an interesting open question whether the protocol can be extended to the case of reference frame alignment. The entanglement of the network state and its usefulness for metrology can be enhanced by LOCC, which depends on the quality of quantum memory. Since the range of entanglement affects the success probability, the relation between the quality of quantum memory and the success probability is also desirable to investigate as the technology develops.

\begin{acknowledgments}
All authors contributed equally.
We thank
Otfried G\"{u}hne, Kiara Hansenne,
Jiaxuan Liu and Sixia Yu
for discussions.
This work was supported 
by Guangdong Basic and Applied Basic Research Foundation (Project No.~2022A1515010340), the Hong Kong Research Grant Council (RGC) through the Early Career Scheme (ECS) grant 27310822 and the General Research Fund (GRF) grant 17303923,
the Deutsche Forschungsgemeinschaft (DFG, German 
Research Foundation, No. 447948357 
and 440958198), the Sino-German Center for 
Research Promotion (Project M-0294), the ERC 
(Consolidator Grant No. 683107/TempoQ) and the 
German Ministry of Education, Research 
(Project QuKuK, BMBF Grant No. 16KIS1618K).
This project has received funding from the European Union’s Horizon 2020 research and innovation programme under the Marie Sk\l{}odowska-Curie grant agreement No.~945422. 
Z.P.X. 
acknowledges support from the Alexander von 
Humboldt Foundation, {National Natural Science Foundation of China} (Grant No.\ 12305007) and 
Anhui Provincial Natural Science Foundation (Grant No.\ 2308085QA29).
\end{acknowledgments}
\bibliographystyle{apsrev4-2}
\bibliography{ref}

\newpage
\appendix

\onecolumngrid

\section{Deterministic distributed sensing with quantum network states}

The performance of a sensing protocol can be characterized by the MSE (mean squared error) $\Delta^2(\theta)$.
In the context of distributed sensing, the Heisenberg limit (HL) refers to the following scaling of the MSE:
\begin{align}
\Delta^2(\theta)\sim M^{-2}
\end{align}
which features a linear improvement over the standard quantum limit $(\Delta\theta)^2\sim M^{-1}$, where $M$ is the number of distinct signals.
A core question of distributed quantum sensing is when is the HL achieved.
In a similar fashion as \cite[Appendix A]{fadel2022multiparameter}, we prove a general lower bound on the error of any deterministic protocol, which implies that any protocol that does not involve genuine multipartite entanglement is dominated by the standard quantum limit.

For a given state $\rho$ and a set of dichotomic observables $\mathcal{A} = \{A_i\}$, denote by $\cov(\rho,\mathcal{A})$ the corresponding covariance matrix, defined as
\begin{equation}
\left[\cov(\rho,\mathcal{A})\right]_{ij}:=\tr(A_iA_j\rho) - \tr(A_i\rho)\tr(A_j\rho).
\end{equation}
For product states, their covariance matrices can be cast into a sum form as stated in the following lemma~\cite{aaberg2020semidefinite} (see the end of this section for the proof): 
\begin{lemma}\label{thm:separable}
For any $K$-partite product state $\rho = \otimes_{t=1}^K \sigma_t$, and an arbitrary set of observables $\mathcal{A} = \{A_i\}_{i=1}^{S}$, the covariance matrix can be expressed as 
\begin{align}\label{eq:covdec}
 & \cov(\rho,\mathcal{A}) = \sum_{k=1}^K \Upsilon^{(k)},
\end{align}
where $\{\Upsilon^{(k)}\}_{k=1}^K$ is a set of positive semi-definite matrices of dimension $S$.
Moreover, for any $1\le i,j\le S$, if either $A_i$ or $A_j$ acts trivially on the $k$-th subsystem,  we have $\Upsilon^{(k)}_{i,j}=0$.
\end{lemma}

We now proceed to bound the performance of distributed sensing within a network. Given a network $G(\mathcal{V},\mathcal{E})$, any sensing protocol consists of generating a probe state $\rho$, acquiring signals, and generating an estimate.
The probe state $\rho$ is generated from the original network state via local operations.
We consider the general setting with multiple signals $\{\theta_s\}_{s\in\mathcal{S}}$.
Each signal is imprinted on $\rho$ via a unitary $e^{-i\theta_s H_s}$, where the generator $H_s$ acts non-trivially on vertices in the subset $s\subset\mathcal{E}$. 
Finally, the probe state is measured to generate an estimate, whose accuracy is constrained by the quantum Fisher information matrix  $\mathcal{F}_Q(\rho, \{H_s\})$  of the state   $\rho_{\vec{\theta}}:=U(\vec{\theta})\rho U(\vec{\theta})^\dag$ with $U(\vec{\theta}):=e^{-i\sum_s H_s\theta_s}$ (cf.~Ref.~\cite{liu2020quantum}). Concretely, the matrix element $\left(\mathcal{F}_Q(\rho,\{H_s\})\right)_{ss'}$ is defined by $\left(\mathcal{F}_Q\right)_{ss'}:=\Tr[\rho_{\vec{\theta}}(L_{s}L_{s'}+L_{s'}L_s)/2]$, with $\{L_s\}$ being the symmetric logarithmic derivatives obtained by solving $\partial\rho_{\vec{\theta}}/\partial \theta_s=(L_s\rho_{\vec{\theta}}+\rho_{\vec{\theta}}L_s)/2$.

With Lemma~\ref{thm:separable} we show a crucial feature of $\mathcal{F}_Q(\rho, \{H_s\})$:
\begin{theo}\label{thm:fisherdec}
There exists a set of $|\mathcal{S}|$-dimensional square matrices $\{T^{(e)}\}_{e\in \mathcal{E}}$ (whose indices are elements in $\mathcal{S}$) such that
\begin{align}\label{eq:sdpdec}
    &\sum_{e\in  \mathcal{E}} T^{(e)} \ge \mathcal{F}_Q(\rho, \{H_s\})/4,\nonumber\\
    &\sum_{e\in \mathcal{E}} T^{(e)}_{ss} = \var(\rho,H_s),\nonumber\\
    &\Pi^{(e)} T^{(e)} \Pi^{(e)} = T^{(e)},\  T^{(e)} \ge 0,\  \forall e\in \mathcal{E},
\end{align}
where $\var(\rho, H_s)$ is the variance of $H_s$ with respect to  $\rho$, $\Pi^{(e)} = \sum_{s\cap e \neq \emptyset} P^{(s)}$, and $P^{(s)}$ is the projector onto the block for the subset  $s$.
\end{theo}
\begin{proof}
By definition, the probe state $\rho$ has the general-form decomposition
\begin{align}
    \rho &= \sum_\lambda p_\lambda \rho_\lambda,\\ 
    \rho_\lambda &= \left(\bigotimes_{v\in \mathcal{V}} \Phi_{v}^{(\lambda)} \right) \left(\bigotimes_{e\in \mathcal{E}} \sigma_{e}\right),\label{eq:nolosr}
\end{align}
where $\Phi_{v}^{(\lambda)}$ is a channel acting on sensor $v$.
If all three (in)equalities in Eq.~\eqref{eq:sdpdec} hold for each $\rho_\lambda$, they hold also for $\rho$ by the convexity of $\mathcal{F}_Q$ and the concavity of variance~\cite{Toth2013Extremal}. Hence, we focus on  $\rho$ of the form \eqref{eq:nolosr} without loss of generality.  

Next, note that each channel $\Phi_v$ has the Stinespring dilation:
\begin{equation}
  \Phi_v(\tau) = \tr_{R(v)}[U_v(\tau\otimes (\omega_{v})_{R(v)})U_v^\dagger], 
\end{equation}
where $\omega_{v}$ is an ancillary state and $U_v$ is a unitary acting on system $v$ and its ancilla $R(v)$.
This implies that
\begin{equation}\label{eq:tilde_rho}
    \rho = \tr_{\mathcal{R}}[U_{\mathcal{V}} \tilde{\rho} U_{ \mathcal{V}}^\dagger],\  \tilde{\rho} = \left(\bigotimes_{e\in \mathcal{E}} \sigma_e\right)\otimes \left(\bigotimes_{v\in \mathcal{V}} \omega_{v}\right),
\end{equation}
where $\mathcal{R}$ denotes the collection of all ancillas and $U_{\mathcal{V}} = \otimes_{v\in \mathcal{V}} U_v$.
Defining $\tilde{H}_s := U_{\mathcal{V}}^\dagger (H_s \otimes \id_{\mathcal{R}}) U_{\mathcal{V}}$, we have $\tr(H_s \rho) = \tr(\tilde{H}_s \tilde{\rho}),\ \tr(H_sH_t \rho) = \tr(\tilde{H}_s \tilde{H}_t \tilde{\rho})$ for any $s,t\in\mathcal{S}$,
which implies that
\begin{equation}\label{eq:coveq}
    \cov(\rho, \{H_s\}) = \cov(\tilde{\rho},\{\tilde{H}_s\}).
\end{equation}

To prove Eq.~(\ref{eq:sdpdec}), we first show that there exists a set of $\{T^{(e)}\}$ that decomposes the covariance matrix $\cov(\rho, \{H_s\})$ and then apply the relation between the covariance matrix and the quantum Fisher information matrix.
Since $\tilde{\rho}$, as defined in Eq.~(\ref{eq:tilde_rho}), is a product state,  Lemma~\ref{thm:separable} combined with Eq.~(\ref{eq:coveq}) imply the decomposition
\begin{align}\label{eq:covdecthm}
     \cov(\rho, \{H_s\})  = \sum_{e\in \mathcal{E}} \Upsilon^{(e)} + \sum_{v\in \mathcal{V}} \Upsilon^{(v)}, 
\end{align}
where $\Upsilon^{(e)}$ and  $\Upsilon^{(v)}$ satisfy the properties in Eq.~\eqref{eq:covdec}. 
We further define 
\begin{align}
T^{(e)} := \Upsilon^{(e)} + \sum_{v\in e} \Upsilon^{(v)}/c_v,
\end{align}
where $c_v$ is the number of hyperedges containing the vertex $v$. By definition, the set $\{T^{(e)}\}_{e\in \mathcal{E}}$ satisfies the third equality in Eq.~\eqref{eq:sdpdec}. In addition, as 
\begin{align}
 \cov(\rho, \{H_s\})=\sum_{e\in\mathcal{E}} T^{(e)},
\end{align}
the first inequality in Eq.~\eqref{eq:sdpdec} follows from the relation between the quantum Fisher information matrix and the covariance matrix: $\mathcal{F}_Q(\rho, \{H_s\}) \le4\cov(\rho, \{H_s\}) $ \cite{gessner2018sensitivity}, and the second equality follows from the definition of the covariance matrix.  
\end{proof}

Let us now consider the task of estimating a parameter that is an arbitrary linear combination of the local parameters:
\begin{align}
\theta(\vec{\alpha}):=\vec{\alpha}^T\vec{\theta}=\sum_s \alpha_s\theta_s.
\end{align}
The mean squared error can be expressed as
\begin{align}\label{mse-theta}
\Delta^2\left(\theta(\vec{\alpha})\right)=\vec{\alpha}^T\cov\left(\{\theta_s\}\right)\vec{\alpha},
\end{align}
where $\cov\left(\{\theta_s\}\right)$ is the covariance matrix of the parameters $\{\theta_s\}$. 
Defining $\tilde{H}_s:=\alpha_s H_s$, it is immediate that $\vec{\alpha}^T \mathcal{F}_Q(\rho,\{H_s\})\vec{\alpha}$ equals the sum of all entries in the matrix $\mathcal{F}_Q(\rho,\{\tilde{H}_s\})$.
Applying Lemma~\ref{thm:separable} to $\{\tilde{H}_s\}$, we get
\begin{align}
    \vec{\alpha}^T \mathcal{F}_Q(\rho,\{H_s\})\vec{\alpha} &\le 4 \sum_{st} \sum_e \tilde{T}^{(e)}_{st},
\end{align} 
where $\{\tilde{T}^{(e)}\}_{e\in\mathcal{E}}$ are positive semi-definite matrices satisfying Eq.~(\ref{eq:sdpdec}) (with $\{H_s\}$ substituted by $\{\tilde{H}_s\}$).
The inequality can be further expanded as
\begin{align}
   \vec{\alpha}^T \mathcal{F}_Q(\rho,\{H_s\})\vec{\alpha}  &\le 4 \sum_e \sum_{st} \sqrt{\tilde{T}^{(e)}_{ss} \tilde{T}^{(e)}_{tt}}\\
   & \le 4 \sum_e \sum_{st} \frac{1}{2} \left( \tilde{T}^{(e)}_{ss} + \tilde{T}^{(e)}_{tt} \right) \\
    &\le 4 \sum_e  k^{(e)} \sum_{s} \tilde{T}^{(e)}_{ss}\\
    &\le 4 \sum_s k_s \sum_e \tilde{T}^{(e)}_{ss}\\
    &= 4 \sum_{s} k_s \var(\rho,\tilde{H}_s)\\
    &= 4 \sum_{s} k_s \alpha_s^2\var(\rho,H_s),
\end{align}
where $k^{(e)}$ is the
number of $s$ for which $\tilde{T}^{(e)}_{ss} > 0$
and 
$k_s = \max_{e\in \mathcal{E}, e\cap s\neq \emptyset} |\{t \in \mathcal{S} |t\cap e\neq \emptyset\}|$.  Note that when all the Hamiltonians act locally without interaction, $k_s$ is the maximal size of edge containing the node $s$.
The first inequality follows from  $\tilde{T}^{(e)} \geq 0$, the second from the arithmetic--geometric mean inequality, and the next two steps from the definitions of $k^{(e)}$ and $k_s$.
The final two lines are from the second equation in \eqref{eq:sdpdec} and the definition of $\tilde{H}_s$.
The last inequality can be written as
\begin{align}
    \vec{\alpha}^T \mathcal{F}_Q \vec{\alpha} & \leq \vec{\alpha}^T \Phi(V) \vec{\alpha}, \\
    \Phi(V) & :=  \diag \{ 4k_s \var(\rho,H_s) \}_s.
\end{align}
As this holds for all $\vec{\alpha}$, we have $\mathcal{F}_Q \le \Phi(V)$.
Taking the inverse,
\begin{align}
    \mathcal{F}_Q^{-1} \ge \Phi(V) ^{-1} \ge \diag \left\{ \frac{1 }{4k_s \var(\rho,H_s)} \right\}_s,
\end{align}
using the fact that the function $f(t) = -t^{-1}$ is operator monotone.
The matrix Cram\'er-Rao bound \cite{helstrom1969quantum,holevo2011probabilistic} then implies
\begin{align}
   \cov\left(\{\hat{\theta}_s\}\right) \ge (1/\nu)\diag \Phi(V)^{-1} = \diag \left\{ \frac{1 }{4\nu k_s \var(\rho,H_s)} \right\}_s
\end{align}
where $\nu$ is the number of repetitions. For a given $\vec{\alpha}$, this results in the bound
\begin{align}
    \Delta^2\left(\theta(\vec{\alpha})\right)=\vec{\alpha}^T\cov\left(\{\hat{\theta}_s\}\right)\vec{\alpha} \geq \sum_s \frac{ \alpha_s^2}{4\nu k_s \var(\rho,H_s)}.
\end{align}

\medskip
\noindent{\it Proof of Lemma \ref{thm:separable}.}~
We shall prove the lemma by explicitly constructing $\{\Upsilon^{(k)}\}$. Define
\begin{align} 
 & \cov(\rho,\mathcal{A}) = \sum_{k=1}^K \Upsilon^{(k)},\\ 
 & \Upsilon^{(k)}_{ij} = \tr_{k\to}[(A_i^{(k-1)} - \id_k \otimes A_i^{(k)}) (A_j^{(k-1)} - \id_k \otimes A_j^{(k)}) \rho_{k\to} ]
\end{align}
where $k\to := (k,\ldots,K)$, $\ot k := (1,\ldots,k)$,
\begin{align}\label{eq:decomposition_procedure}
  &A_i^{(k)} = \tr_{\ot k} (A_i\rho_{\ot k}),\ A_i^{(0)} = A_i,\nonumber\\
  &\rho_{S} = \otimes_{t \in S} \sigma_t,\ \rho_{K+1\to} = 1.
\end{align}

Since $\tr[A_j^{(k-1)} \rho_{k\to}]=\tr[(\id_k \otimes A_j^{(k)}) \rho_{k\to}]=\tr[A_j\rho]$, $\Upsilon^{(k)}$ is the covariance matrix of a set of observables $\{A_j^{(k-1)} - \id_k \otimes A_j^{(k)}\}$. Therefore, we have 
\begin{align}
\Upsilon^{(k)}\ge0.
\end{align}

Next, we prove that these matrices sum up to the covariance matrix.
Direct calculation shows that
\begin{align}\label{eq:transform}
  \Upsilon^{(k)}_{ij} =& \tr[(A_i^{(k-1)} - \id_k \otimes A_i^{(k)}) (A_j^{(k-1)} - \id_k \otimes A_j^{(k)}) \rho_{k\to} ]\\
  =& \tr[A_i^{(k-1)} A_j^{(k-1)} \rho_{k\to} ] - \tr[(\id_k \otimes A_i^{(k)})A_j^{(k-1)} \rho_{k\to} ] - \tr[A_i^{(k-1)} (\id_k \otimes A_j^{(k)}) \rho_{k\to} ]\\ 
                      &+ \tr[(\id_k \otimes A_i^{(k)}) (\id_k \otimes A_j^{(k)}) \rho_{k\to} ]\\
  =& \tr[A_i^{(k-1)} A_j^{(k-1)} \rho_{k\to} ] - \tr[A_i^{(k)} A_j^{(k)} \rho_{k+1\to} ],
\end{align}
where the first equality is from the definition, the second equality is from expansion, and the third equality is from the facts that
\begin{align}
  \tr_{k}[(\id_k \otimes A_i^{(k)}) A_j^{(k-1)} \rho_{k\to} ] = A_i^{(k)}A_j^{(k)} \rho_{k+1\to},\\
  \tr_{k}[A_i^{(k-1)} (\id_k \otimes A_j^{(k)}) \rho_{k\to} ] = A_i^{(k)}A_j^{(k)} \rho_{k+1\to},\\
  \tr_{k}[(\id_k \otimes A_i^{(k)}) (\id_k \otimes A_j^{(k)}) \rho_{k\to} ] = A_i^{(k)}A_j^{(k)} \rho_{k+1\to},
\end{align}
having used the following property of the partial trace: $
  \tr_X[(A_X\otimes A_Y)\rho_{XY}] = A_Y \tr_X[(A_X\otimes \id_Y) \rho_{XY}]$.
From Eq.~\eqref{eq:transform}, we have
\begin{align}
  \sum_{k=1}^K \Upsilon^{(k)} &= \tr[A_i^{(0)} A_j^{(0)} \rho_{1\to} ] - \tr[A_i^{(K)} A_j^{(K)} \rho_{K+1\to} ]\\
                              &= \tr[A_i A_j \rho ] - A_i^{(K)} A_j^{(K)} \\
                              &= \cov(\rho,\mathcal{A}),
\end{align}
where the first equality is from direct summation, the second equality is by definition, the third equality is from the fact that $A_i^{(K)} = \tr(A_i \rho)$.

If $A_i$ acts on the $k$-th party trivially, i.e., $A_i = \id_k \otimes \tilde{A}_{i,\bar{k}}$ for some $\tilde{A}_{i,\bar{k}}$, then
\begin{align}
A_i^{(k-1)} &= \tr_{\ot k-1} (A_i\rho_{\ot k-1})\nonumber\\
            &= \tr_{\ot k-1} [(\id_k\otimes\tilde{A}_{i,\bar{k}}) \rho_{\ot k-1}]\\
            &= \id_k\otimes\tr_{\ot k-1} [\tilde{A}_{i,\bar{k}} \rho_{\ot k-1}]\\
            &= \id_k\otimes\tr_{\ot k} [(\id_k\otimes\tilde{A}_{i,\bar{k}}) (\sigma_k\otimes\rho_{\ot k-1})]\\
            &= \id_k\otimes\tr_{\ot k} [A_i \rho_{\ot k}]\\
            &= \id_k\otimes A_i^{(k)}.
\end{align}
This immediately leads to $\Upsilon_{ij}^{(k)} = 0$ when at least one of $\{A_i,A_j\}$ acts trivially on the subsystem $k$.
\qed

\section{Bounds for shallow circuits when the gates contain the parameter of interest} 

A similar setting involves a circuit where the gates themselves now contain the parameter of interest.
That is, we take gates $U_{j,\alpha}(\theta)$ such that $\dot U_{j,\alpha}(\theta) := \partial_\theta U_{j,\alpha}(\theta) = -i H_{j,\alpha} U_{j,\alpha}(\theta)$.
To simplify things, let us define $U_j = \otimes_\alpha U_{j,\alpha}$ and $H_j = \sum_\alpha H_{j,\alpha}$.
First note that
\begin{align}
    i\dot U & =  U_D \dots U_2 H_1 U_1 + U_D \dots U_3 H_2 U_2 U_1 + \dots + H_D U_D \dots U_1 \nonumber \\
        & = \sum_{j=1}^D U_{>j} H_j U_{\leq j} \nonumber \\
        & = \sum_{j=1}^D U_{>j} H_j U_{>j}^\dagger U_{>j} U_{\leq j} \nonumber \\
        & = H' U,
\end{align}
where $H' := \sum_j U_{>j} H_j U_{>j}^\dagger$.
Therefore, the derivative of the output state $\sigma = U\rho U^\dagger$ is $\dot \sigma = \dot U \rho U^\dagger + U \rho \dot U^\dagger = -i(H' U\rho U^\dagger - U\rho U^\dagger H') = -i[H',\sigma]$.
Using the Heisenberg picture, we have $\mathcal{F}_Q(\sigma,H') = \mathcal{F}_Q(\rho, \tilde H)$, with
\begin{align}
    \tilde H & := U^\dagger H' U \nonumber \\
        & = \sum_{j=1}^D U^\dagger U_{>j} H_j U_{>j}^\dagger U \nonumber \\
        & = \sum_{j=1}^D U_{\leq j}^\dagger H_j U_{\leq j} \nonumber \\
        & =: \sum_{j=1}^D \tilde{H}_j.
\end{align}
Overall, the relevant Hamiltonian $\tilde{H}$ thus has terms $U_{\leq j}^\dagger H_{j,\alpha} U_{\leq j}$ whose locality is bounded by the structure and depth of the circuit.
Again taking a one-dimensional chain with 2-site nearest neighbor interactions and an initial separable state, we first bound the number of terms $\tilde{H}_{j',\alpha'}$ that overlap with a given $\tilde{H}_{j,\alpha}$.
The weight of $\tilde{H}_{j,\alpha}$ is at most $q_j \leq 2(j+1)$ (over an interval of at most this size).
The number of overlapping terms from layer $j'$ is thus at most $q_j + q_{j'} - 1 \leq 2j + 2j' +3$.
Summing over all $j'$, we have
\begin{align}
    \text{number overlapping } \tilde{H}_{j,\alpha} & \leq \sum_{j'=1}^D 2j + 2j' + 3 \nonumber \\
        & = 2Dj + D(D+1) + 3D \nonumber \\
        & = 2Dj + D^2 + 4D.
\end{align}
Therefore, the QFI is bounded by
\begin{align}
    \mathcal{F}_Q(\rho, \tilde{H}) & \leq 4 \sum_{j=1}^D D(2j+D+4) \sum_\alpha V(\rho, \tilde{H}_{j,\alpha}).
\end{align}

\section{Achieving the Heisenberg limit with post-selection}
Consider a generic connected hypergraph $G=(V,E)$. We assume without loss of generality that $G$ is simple, i.e., every edge is unique, and does not have any self-loop, i.e., an edge containing only one vertex, which corresponds to a trivial, non-entangling source.

In Protocol \ref{protocol-general}, observe that the operations starting from Step 2 (center election and local pre-processing) are done locally and their order of execution within each sensor is arbitrary (once $\{\tilde{\alpha}_j\}$ are distributed). 
For convenience of our proof, we consider the operations (signal acquisition, measurement, and post-selection) sensor by sensor, and delay all operations at the center to the end.

Suppose a non-center sensor $v_0$ is measured first. Conditioning on the post-selection being successful, the signal of $v_0$ will be passed to an entangled state shared by the neighbors of $v_0$ (i.e., sensors entangled with $v_0$), which we call the signal state.
We denote by $S$ the sensors who have finished their local measurements and by $\tilde{S}\subset V$ the sensors sharing the signal state. 
At this stage, we have $S=\{v_0\}$ and $\tilde{S}=V_{\rm neigh}(v)$, where
\begin{align}
    V_{\rm neigh}(v):=\left\{v'\in V~: (v, v')\in E\right\}.
\end{align}

Moving on, as long as there exists a non-center $v\in\tilde{S}$, we shall proceed by considering the measurement of $v$.
It can be shown that, conditioning on the measurement being successful, we have:
\begin{itemize}
    \item the set of sensors who completed their local measurements is updated to $S'=S\cup\{v\}$;
    \item the signal state is now an entangled state (i.e., a GHZ state carrying a phase) shared by $\tilde{S}'=(\tilde{S}\setminus\{v\}) \cup (V_{\rm neigh}(v)\setminus S)$.
    \item denoting a qubit of the network state by its site $v\in V$ and source $e\in E$ as $[e,v]$, the signal state at this stage becomes $|\Psi(Q_{S',\tilde{S}'},S')\>$, where
\begin{align}
    |\Psi(Q,S)\>:=\frac{1}{\sqrt{2}}\left(|00\cdots 0\>_{Q}+e^{i\sum_{j:v_j\in S}\tilde{\alpha}_j\theta_j}|11\cdots1\>_{Q}\right).
\end{align}
Here $|00\cdots 0\>_Q,|11\cdots 1\>_Q$ are states on qubits in the set qubits $Q$, and 
\begin{align}
    Q_{S,\tilde{S}}:=\left\{[e,v]\in E\times\tilde{S}~:~\exists v'\in S,\ \text{s.t.}\  (v,v')= e\right\}.
\end{align}
\end{itemize} 

We show the above claim by induction.
Define the set of edges $$E_S(v):=\left\{e\in E~:~v\in e,(S\setminus \{v\})\cap e=\emptyset\right\}.$$ Right before the measurement in Step 5 (Measurement and post-selection), 
the first sensor $v_0$ holds plus states that carry exactly $|\tilde{\alpha}_{v_0}|$ queries to the local signal $\theta_{v_0}$ of $v_0$, and one qubit of the GHZ state $|\GHZ_e\>$ for every $e\in E_S(v_0)$.
Straightforward calculation shows that, when $\tilde{\alpha}_{v_0}>0$, measuring all the above states and post-selecting on the GHZ state of all qubits would entangle all other qubits  in $|\GHZ_e\>$ that are not held by $v_0$ for every $e\in E_S(v_0)$, creating a GHZ state that carries a phase $\tilde{\alpha}_{v_0}\theta_{v_0}$. When $\tilde{\alpha}_{v_0}<0$, applying $X$ to each plus state will effectively flip the sign of the phase and thus ensure that the global phase  is of the correct sign.

At an intermediate stage, for any $S,\tilde{S}$ and any non-center $v\in\tilde{S}$, we have a similar argument:
Locally, the sensor $v$ holds the signal state $|\Psi(Q_{S,\tilde{S}},S)\>$,  plus states that carry exactly $|\tilde{\alpha}_v|$ queries to the local signal $\theta_v$ of $v$, and one qubit of the GHZ state $|\GHZ_e\>$ for every $e\in E_S(v)$.
It is straightforward that, when $\tilde{\alpha}_v>0$, measuring all the above states and post-selecting on the GHZ state of all qubits would pass an overall phase $\tilde{\alpha}_v\theta_v$ to $|\Psi(Q_{S,\tilde{S}},S)\>$, entangle it to $|\GHZ_e\>$ for every $e\in E_S(v)$ while disentangling it with all qubits of $v$. When $\tilde{\alpha}_v<0$, applying $X$ to each plus state will effectively flip the sign of the phase and thus ensure that the global phase $\tilde{\alpha}_v\theta_v$ is passed to the signal state with the correct sign.
Either way, the new signal state is $|\Psi(Q_{S',\tilde{S}'},S')\>$, with $S'$ and $\tilde{S}'$ being the updated $S$ and $\tilde{S}$, respectively, as described by the rules.

One can see that the above procedure converges to $S=V\setminus\{v^\ast\}$ and $\tilde{S} = \{v^\ast\}$. Otherwise, there are two possible cases:
\begin{enumerate}
    \item The procedure terminates with $\tilde{S}=\emptyset$.
    \item The procedure terminates with $\tilde{S}=\{v^\ast\}$ and $S\cup\tilde{S}\not=V$.
\end{enumerate}
If the first case is true, it means that $v^\ast$ has never been added to $\tilde{S}$ and $S$. Since $G$ is connected, there is a path from $v_0$ to $v^\ast$. Since $v_0\in S$ and $v^\ast\not\in S$, there exists two adjacent vertices $v_1$ and $v_2$ along the path such that $v_1\in S$ and $v_2\not\in S$. However, when $v_1$ was measured, all its neighbors, including $v_2$, were added to $\tilde{S}$. If $v_2\not=v^\ast$ and $v_2\in\tilde{S}$, it would be measured and added to $S$. The only exception is $v_2=v^\ast$, which reduces to the second case.

If the second case is true, there exists $v\in V\setminus (S\cup\tilde{S})$. Since $v^\ast$ is not a cut-vertex, there exists a path connecting $v$ with $v_0$ that does not involve $v^\ast$. As $v\not\in S$ and $v_0\in S$, there exist two adjacent vertices $v_1$ and $v_2$ along the path such that $v_1\in S$ and $v_2\not\in S$. Again, when $v_1$ was measured, all its neighbors, including $v_2$, were added to $\tilde{S}$. Since $v_2\not=v^\ast$ and $v_2\in\tilde{S}$, it should have been measured later and added to $S$, and thus we have a contradiction with the assumption that $v_2\not\in S$. In conclusion, neither of the above two cases is possible.

When the procedure converges, the signal state becomes $|\Psi(Q_{V\setminus\{v^\ast\},\{v^\ast\}},V\setminus\{v^\ast\})\>$, a GHZ state at $v^\ast$ that carries a global phase 
$
\sum_{j:v_j\not=v^\ast}\tilde{\alpha}_j\theta_j$. 
By Step 3 (Signal acquisition) and Step 5 (Measurement and post-selection), the center $v^\ast$ holds only $|\Psi(Q_{V\setminus\{v^\ast\},\{v^\ast\}},V\setminus\{v^\ast\})\>$. 
After passing any of the qubits through   $e^{-i\tilde{\alpha}^*\theta^*Z/2}$ (or $Xe^{-i\tilde{\alpha}^*\theta^*Z/2}X$ if $\tilde{\alpha}_j<0$), the signal state becomes a GHZ state carrying the desired global phase $M\theta$. 
The probability of getting $|\GHZ^*\>\<\GHZ^*|$ when measuring it with $\{|\GHZ^*\>\<\GHZ^*|,I-|\GHZ^*\>\<\GHZ^*|\}$ is thus given by
\begin{align}
    P_\theta  = \cos^2(M\theta/2).
\end{align}
The Fisher information is $4M^2$, and thus the HL is achieved, conditioned on the success of the protocol.

\section{Success probability of Protocol \ref{protocol-general}}\label{app:successprob}
  We remark that the success of the complete protocol has the probability $p_{{\rm succ}}$, where
\begin{equation}
    \log_2 p_{{\rm succ}} = -\Big[\sum_{v\neq v^\ast} |\tilde{\alpha}_j| + \sum_{e\in \mathcal{E}} |e| - |\mathcal{E}(v^\ast)| \Big],
\end{equation}
where $\mathcal{E}(v^\ast)$ is the set of edges containing $v^\ast$, $|e|$ is the number of vertices contained in $e$.

First, the probability of success at the first step is $2^{-n_{v_0}}$, where $n_{v_0}$ is the number of particles which have been measured after all. To be more precise, 
the state of the particles related to the node $v_0$ before the measurement is
\begin{equation}
    \Big(\bigotimes_{e\in \mathcal{E}(v_0)} |{\rm GHZ}\rangle_e\Big) \otimes |+_{\theta_{v_0}}\rangle^{\otimes |\tilde{\alpha}_{v_0}|} = 2^{-(|\mathcal{E}(v_0)| + |\tilde{\alpha}_{v_0}|)} \Big( |\mathbf{0}00\ldots 0\rangle + e^{i\tilde{\alpha}_{v_0} \theta_{v_0}} |\mathbf{1}11\ldots 1\rangle +\ \text{rest terms} \Big)
\end{equation}
in the case that $\tilde{\alpha}_{v_0} > 0$, where $\mathcal{E}(v_0)$ is the set of edges containing vertex $v_0$ in the network, $|\mathbf{0}\rangle = |00\ldots 0\rangle$ is the state for the qubits distributed to other vertices, etc. The remaining terms do not contain $|00\ldots 0\rangle, |11\ldots 1\rangle$ for the qubits belonging to the node $v_0$. Hence, the probability of success to project this state onto the GHZ state $(|00\ldots 0\rangle + |11\ldots 1\rangle)/\sqrt{2}$ for the qubits belonging to the node $v_0$ is $2^{-n_{v_0}}$ with $n_{v_0} = |\mathcal{E}(v_0)| + |\tilde{\alpha}_{v_0}|$ to be the number of qubits  measured on the node $v_0$. This result follows also similarly for the case where $\tilde{\alpha}_{v_0} < 0$.

After the measurement at each step and the post-selection on the successful case, we obtain in fact a new network as discussed before, where some GHZ states carry nontrivial phases. 
To be more explicit, 
the state of the particles related to the node $v_0$ before the next measurement is
\begin{equation}
    \Big(\bigotimes_{e\in \bar{\mathcal{E}}(v)} |{\rm GHZ}_{\theta_e}\rangle_e\Big) \otimes |+_{\theta_{v}}\rangle^{\otimes |\tilde{\alpha}_{v}|} = 2^{-(|\bar{\mathcal{E}}(v)| + |\tilde{\alpha}_{v}|)} \Big( |\mathbf{0}00\ldots 0\rangle + e^{i \hat{\theta}_v} |\mathbf{1}11\ldots 1\rangle +\ \text{rest terms} \Big),
\end{equation}
where $|{\rm GHZ}_{\theta_e}\rangle = (|00\ldots 0\rangle + e^{i\theta_e}|11\ldots 1\rangle)/\sqrt{2}$, $\hat{\theta}_v = \tilde{\alpha}_v \theta_v + \sum_{e\in \bar{\mathcal{E}}(v)} \theta_e$, and $\bar{\mathcal{E}}(v)$ is the set of edges related to node $v$ in the new network.
One essential point is that $\cup_{e\in \bar{\mathcal{E}}_S(v)} e$ contains qubits for other vertices also, according to our assumption that $v^\ast$ is not a cut-vertex. 
Similarly, the probability of success for an intermediate step corresponding to vertex $v$ is  no less than  $2^{-n_v}$, where $n_v$ is the number of qubits measured on the node $v$. The probability is exactly $2^{-n_v}$ when none of those qubits are in the same GHZ state. 
Notice that the measurement and post-selection on other vertices do not affect the number $n_v$, we have $n_v = |\mathcal{E}(v)| + |\tilde{\alpha}_v|$.

Consequently, the probability $p_{{\rm succ}}$ of the success of the complete protocol satisfies $0 \le -\log_2 p_{{\rm succ}} \le  \sum_{v\neq v^\ast} n_v$, thus,
\begin{equation}
    \log_2 p_{{\rm succ}}  \ge -\sum_{v\neq v^\ast} n_v = -\Big[\sum_{v\neq v^\ast} |\tilde{\alpha}_v| + \sum_{e\in \mathcal{E}} |e| - |\mathcal{E}(v^\ast)| \Big],
\end{equation}
where $\mathcal{E}(v^\ast)$ is the set of edges containing $v^\ast$, and $|e|$ is the number of vertices contained in $e$.

\section{Privacy of the probabilistic protocol}

Consider a hypergraph $G$ such that $G\setminus\{v^\ast\}$ is connected. Here we show that the local parameters $\{\theta_j\}$ will not be disclosed to the center regardless of the outcome of the post-selection. That is, denoting by $S\subset V\setminus\{v^\ast\}$ the set of vertices yielding ``yes" (successful, and other vertices are unsuccessful), the (unnormalized) state 
\begin{align}\label{rho-S}
\rho(S):=\Tr_{\overline{v^\ast}}\left[\rho_{\vec{\theta}}\left(\bigotimes_{v_k\not\in V\setminus(S\cup\{v^\ast\})}(I-|\GHZ_k\>\<\GHZ_k|)\otimes\bigotimes_{v_j\in S}|\GHZ_j\>\<\GHZ_j|\right)\right]
\end{align}
depends only on $\theta(\vec{\alpha})$ and does not depend explicitly on each $\{\theta_j\}$.
When $S=V\setminus\{v^\ast\}$, i.e., when the whole protocol is successful, we have shown that the conditional state at the center is a GHZ state with phase proportional to $\theta$, and the probability of success is independent of $\{\theta_j\}$. What remains are the cases when $S\neq V\setminus\{v^\ast\}$.

Denote by $\rho_{\vec{\theta}}$ the global state after the signals are acquired.
Observe that tracing out $v\in \mathcal{V}$ would decohere any $e\in \mathcal{E}$ incident to it (i.e., $v\in e$) to $(|0\cdots 0\>\<0\cdots 0|+|1\cdots1\>\<1\cdots1|)/2$. Using this property recursively, for any $S\subsetneqq V\setminus\{v^\ast\}$, we get
\begin{align}
\Tr_{\overline{v^\ast}}\left[\rho_{\vec{\theta}}\left(I_{V\setminus(S\cup\{v^\ast\})}\otimes\bigotimes_{v_j\in S}|\GHZ_j\>\<\GHZ_j|\right)\right] = 2^{-\sum_{v_j\in S} (|\mathcal{E}(v_j)| + |\tilde{\alpha}_j|)}\pi_{S},
\end{align}
  where $\mathcal{E}(v)$ is the set of edges containing $v$,
$\Tr_{\overline{v^\ast}}$ denotes tracing out all vertices but the center,
and $\pi_S$ is a (normalized) diagonal mixed state at the center in the energy eigenbasis of $v^\ast$ (which is, as a result, independent of the signal).
Substituting into Eq.~(\ref{rho-S}), we get
\begin{align}
\rho(S)=\sum_{S':S\subset S'}(-1)^{|S'|-|S|}2^{-\sum_{v_j\in S'} (|\mathcal{E}(v_j)| + |\tilde{\alpha}_j|)}\pi_{S'}.
\end{align}
Therefore, $\rho(S)$ is independent of $\{\theta_j\}$ for any $S$.



\end{document}